\documentclass{easychair}
\usepackage{a4wide}
\usepackage{amssymb}
\usepackage{tabularx}
\usepackage{multirow}
\usepackage{tikz}
\usetikzlibrary{positioning}

\title{A note on the complexity of {\sc k-metric dimension}}
\author{Yannick Schmitz \and Duygu Vietz \and Egon Wanke}
\institute{Heinrich-Heine-Universit\"at D\"usseldorf, Germany\\ \email{yannick.schmitz@hhu.de, duygu.vietz@hhu.de, egon.wanke@hhu.de}}
\newtheorem{theorem}{Theorem}

\newtheorem{example}{Example}

\authorrunning{Schmitz, Vietz, Wanke}
\titlerunning{The complexity of {\sc k-metric dimension}}

\newcommand{\DP}[3]{
\begin{center}
\begin{tabularx}{.75\textwidth}{lL}
\hline\hline
\multicolumn{2}{c}{\sc{#1}} \\
\hline
\em{Instance:}& #2\\
\em{Question:}& #3 \\ 
\hline\hline
\end{tabularx}
\end{center}
}

\newcolumntype{L}{>{\raggedright\arraybackslash}X}
\newcolumntype{R}{>{\raggedleft\arraybackslash}X}
\newcolumntype{C}{>{\centering\arraybackslash}X}

\begin{document}

\maketitle

\begin{abstract}
Two vertices $u, v \in V$ of an undirected connected graph $G=(V,E)$ are {\em resolved} by a vertex $w$ if the distance between $u$ and $w$ and the distance between $v$ and $w$ are different. A set $R \subseteq V$ of vertices is a {\em $k$-resolving set} for $G$ if for each pair of vertices $u, v \in V$ there are at least $k$ distinct vertices $w_1,\ldots,w_k \in R$ such that each of them resolves $u$ and $v$. The {\em $k$-Metric Dimension} of $G$ is the size of a smallest $k$-resolving set for $G$. The decision problem {\sc $k$-Metric Dimension} is the question whether G has a $k$-resolving set of size at most $r$, for a given graph $G$ and a given number $r$. In this paper, we proof the NP-completeness of {\sc $k$-Metric Dimension} for bipartite graphs and each $k \geq 2$.
\end{abstract}

%------------------------------------------------------------------------------
\section{Introduction}

The metric dimension of graphs has been introduced in the 1970s independently by Slater \cite{Sla75} and by Harary and Melter \cite{HM76}. We consider simple undirected and connected graphs $G=(V,E)$, where $V$ is the set of vertices and $E \subseteq \{\{u,v\}\,|\, u,v \in V, u \not=v \}$ is the set of edges. Such a graph {\em has metric dimension} at most $r$ if there is a vertex set $R \subseteq V$ such that $|R| \leq r$ and $\forall u,v \in V$, $u \not=v$, there is a vertex $w \in R$ such that $d(w,u) \not= d(w,v)$, where $d(u,v)$ is the distance (the length of a shortest path in an unweighted graph) between $u$ and $v$. The {\em metric dimension} of $G$ is the smallest integer $r$ such that $G$ has metric dimension at most $r$.

If $d(w,u) \not= d(w,v)$, for three vertices $u,v,w$, we say that $u$ and $v$ are {\em resolved} or {\em distinguished} by vertex $w$. If every pair of vertices is resolved by at least one vertex of a vertex set $R$, then $R$ is a {\em resolving set} or {\em metric generator} for $G$. In certain applications, the vertices of a resolving set are also called {\em landmark nodes} or {\em anchor nodes}. This is a common naming, particularly in the theory of sensor networks. 

The metric dimension finds applications in various areas, including network discovery and verification \cite{BEEHHMR05}, geographical routing protocols \cite{LA06}, combinatorial optimization \cite{ST04},
sensor networks \cite{HW12}, robot navigation \cite{KRR96} and chemistry \cite{CEJO00,Hay17}.

There are several algorithms for computing a minimum resolving set in polynomial time for special classes of graphs, for example trees \cite{CEJO00,KRR96}, wheels \cite{HMPSCP05}, grid graphs \cite{MT84}, $k$-regular bipartite graphs \cite{BBSSS11}, amalgamation of cycles \cite{IBSS10} and outerplanar graphs \cite{DPSL12}. The approximability of the metric dimension has been studied for bounded degree, dense and general graphs in \cite{HSV12}. Upper and lower bounds on the metric dimension are considered in \cite{CGH08,CPZ00} for further classes of graphs.

In this paper, we consider the {\em $k$-Metric Dimension} for some positive integer $k$. A set $R \subseteq V$ of vertices is a {\em $k$-resolving set} for $G$ if for each pair of vertices $u,v \in V$ there are at least $k$ vertices $w_1,\ldots,w_k \in R$ such that each of them resolves $u$ and $v$. The {\em $k$-Metric Dimension} of $G$ is the size of a smallest $k$-resolving set for $G$. The {\sc $k$-Metric Dimension} problem was introduced by Estrada-Moreno et al. in \cite{ERY13}. The $1$-metric dimension is simply called metric dimension. The $2$-metric dimension is also called {\em fault-tolerant metric dimension} and was introduced in \cite{HMSW08}.

Estrada-Moreno et al. analysed the {\sc $(k, t)$-Metric Dimension} \cite{EYR16}. The {\sc $(k, t)$-Metric Dimension} is the {\sc $k$-Metric Dimension}, with the addition, that the distance between two vertices $u, v$ of $G$ is defined as the minimum of $d(u,v)$ and $t$. Therefore, if $t$ is set to the diameter of $G$, the {\sc $(k, t)$-Metric Dimension} is the same as the {\sc $k$-Metric Dimension}. Estrada-Moreno et al. showed the NP-completeness of {\sc $(k, t)$-Metric Dimension} for odd values of $k$.

The decision problem {\sc $k$-Metric Dimension} is defined as follows.

\DP
{$k$-Metric Dimension}
{An undirected connected graph $G=(V,E)$ and a number~$r$.}
{Is there a $k$-resolving set $R \subseteq V$ for $G$ of size at most $r$?}

The complexity of deciding {\sc $k$-Metric Dimension} has only been investigated for very few graph classes, such as trees and other simple graph classes. For general graph classes, {\sc $k$-Metric Dimension} is assumed to be NP-complete if $k$ is given as part of the input. The decision problem {\sc 1-Metric Dimension} is known to be NP-complete, see \cite{GJ79}. A proof can be found in \cite{KRR96}. In this paper, we show the NP-completeness of {\sc $k$-Metric Dimension} for bipartite graphs and each $k \geq 2$ by an alternative approach to \cite{YER17}, whose proof unfortunately is incorrect and does not offer any simple correction options.
 
%------------------------------------------------------------------------------
\section{The NP-completeness of {\sc $k$-Metric Dimension}}

In this section, {\sc $k$-Metric Dimension} is shown to be NP-complete for bipartite graphs and each $k \geq 2$ by a reduction from {\sc 3-Dimensional $k$-Matching}, which is defined as follows.

\DP
{3-Dimensional $k$-Matching (3D$k$M)}
{A set $S \subseteq A \times B \times C$, where $A$, $B$ and $C$ are disjoint sets of the same size $n$.}
{Does $S$ contain a $k$-matching, i.e. a subset $M$ of size $k \cdot n$ such that each element of $A$, $B$ and $C$ is contained in exactly $k$ triples of $M$?}

For $k=1$, the {\sc 3D1M} problem is the well-known NP-complete {\sc 3-Dimensional Matching (3DM)} problem, see \cite{GJ79}. The next theorem shows that {\sc 3D$k$M} is also NP-complete for each $k \geq 2$.

\begin{theorem}\label{theorem-3D$k$M}
\label{T01}
{\sc 3D$k$M} is NP-complete for each $k \geq 2$.
\end{theorem}

\begin{proof}
The {\sc 3D$k$M} problem is obviously in NP, because it can be checked in polynomial time whether a selection of triples from $S$ is a $k$-matching.

The NP-hardness is shown by a reduction from {\sc 3DM}. Let 
$$\begin{array}{ll}
A=\{a_1,\ldots,a_n\}, &
B=\{b_1,\ldots,b_n\}, \\
C=\{c_1,\ldots,c_n\}, \text{ and} &
S=\{s_1,\ldots,s_m\} \\
\end{array}$$
be an instance for {\sc 3DM}. Without loss of generality, $n$ is assumed to be a multiple of $(k-1)$, that is $n=r  (k-1)$ for a positive integer $r$. If this is not the case, then expand $A$, $B$ and $C$ by at most $k-2$ elements each and $S$ by at most $k-2$ triples, which cover every additional element exactly once and none of the originally given elements.

Now consider the following instance for {\sc 3D$k$M} defined by
$$\begin{array}{ll}
A'= A \, \cup \, \{a_{n+1},\ldots,a_{3n}\}, &
B'= B \, \cup \, \{b_{n+1},\ldots,b_{3n}\}, \\
C'= C \, \cup \, \{c_{n+1},\ldots,c_{3n}\}, \text{ and} &
S'= S \, \cup \, R \, \cup \, T \\
\end{array}$$
where $R=\{(a_i,b_i,c_i) \, | \, n+1 \leq i \leq 3n\}$ and $T \subseteq A' \times B' \times C'$. Set $T$ is a set with $3  n (k-1)$ triples, which will be defined later.

The set $A'$, $B'$ and $C'$ is the set $A$, $B$ and $C$ respectively, each expanded by additional $2  n$ elements. Set $S'$ is the set $S$ expanded by the $2  n$ triples of $R$ and the $3  n  (k-1)$ triples of $T$.

Let $U = A \cup B \cup C$ and $U' = A' \cup B' \cup C'$. The $2n$ triples of $R$ cover each element of $U' \setminus U$ exactly once and no element of $U$. Set $T$ will be defined such that its $3 n (k-1)$ triples cover each element of $U'$ exactly $k-1$ times. Each triple of $T$ will have exactly one element from $U$ and two elements from $U' \setminus U$.

If $M$ is a matching for $U$ then $M \cup R \cup T$ is obviously a $k$-matching for $U'$ for any $k \geq 2$. Any $k$-matching $M'$ for $U'$ contains all triples from $R$ and $T$, because otherwise it is not possible to cover the elements of $U'\setminus U$ at least $k$ times. The triples of $T$ cover the elements of $U'$ exactly $k-1$ times. That is, if $M'$ is a $k$-matching for $U'$ then $M = M'\setminus (R \cup T)$ is a matching for $U$.

The set $T$ of triples can be easily defined with the help of a set
$$T_{p,q} \, \subseteq \, (A \times B) \, \cup \, (A \times C) \, \cup \, (B \times C)$$ of tuples defined by $$T_{p,q} =
\begin{array}{ll}
      & \{ (a_i,b_j) \, | \, i \in \{p,\ldots,p+q-1\}, \, j \in \{p+q,\ldots,p+2q-1\}\} \\
\cup  & \{ (b_i,c_j) \, | \, i \in \{p,\ldots,p+q-1\}, \, j \in \{p+q,\ldots,p+2q-1\}\} \\
\cup  & \{ (c_i,a_j) \, | \, i \in \{p,\ldots,p+q-1\}, \, j \in \{p+q,\ldots,p+2q-1\}\}. \\
\end{array}
$$ 
These $3 q^2$ tuples cover each element of $$\{ a_{p}, \ldots, a_{p+2q-1}, \, b_{p},\ldots,b_{p+2q-1}, \, c_{p},\ldots,c_{p+2q-1} \}$$ exactly $q$ times. There are 
\begin{itemize}
\item $q^2$ tuples between the elements of $\{a_{p},\ldots,a_{p+q-1}\}$ and  $\{b_{p+q},\ldots,b_{p+2q-1}\}$,
\item $q^2$ tuples between the elements of $\{b_{p},\ldots,b_{p+q-1}\}$ and  $\{c_{p+q},\ldots,c_{p+2q-1}\}$, and
\item $q^2$ tuples between the elements of $\{c_{p},\ldots,c_{p+q-1}\}$ and  $\{a_{p+q},\ldots,a_{p+2q-1}\}$.
\end{itemize}
Now let $T'$ be the set of tuples defined by
$$T' \, = \, \bigcup_{i=0}^{r-1} T_{n+1+i2(k-1), \, k-1}, \, \text{ with } r = \frac{n}{k-1}.$$
$T'$ contains $r 3 (k-1)^2 \, = \, \frac{n}{k-1} \cdot 3 (k-1)^2 \, = \, 3 n (k-1)$ tuples. It is the union of $r = \frac{n}{k-1}$ sets $T_{p,q}$ where index $p$ is running from $n+1$ to $3n+1-2(k-1)$ in steps of width $2(k-1)$ and $q \, = \, k-1$. These tuples of $T'$ cover each element of $U' \setminus U$ exactly $(k-1)$ times.

In the last step, the $3n(k-1)$ tuples of $T'$ are expanded to $3n(k-1)$ triples for $T$, by including each element from $U$ to exactly $k-1$ tuples from $T'$, such that each generated triple is from the set $A' \times B' \times C'$. Each tuple from $T'$ is extended by exactly one element from $U$. The result is the set $T$ of triples with the required properties. This transformation can obviously be done in polynomial time, see also Example \ref{E01}.
\end{proof}

\begin{figure}
\center
\includegraphics[width=0.90\textwidth]{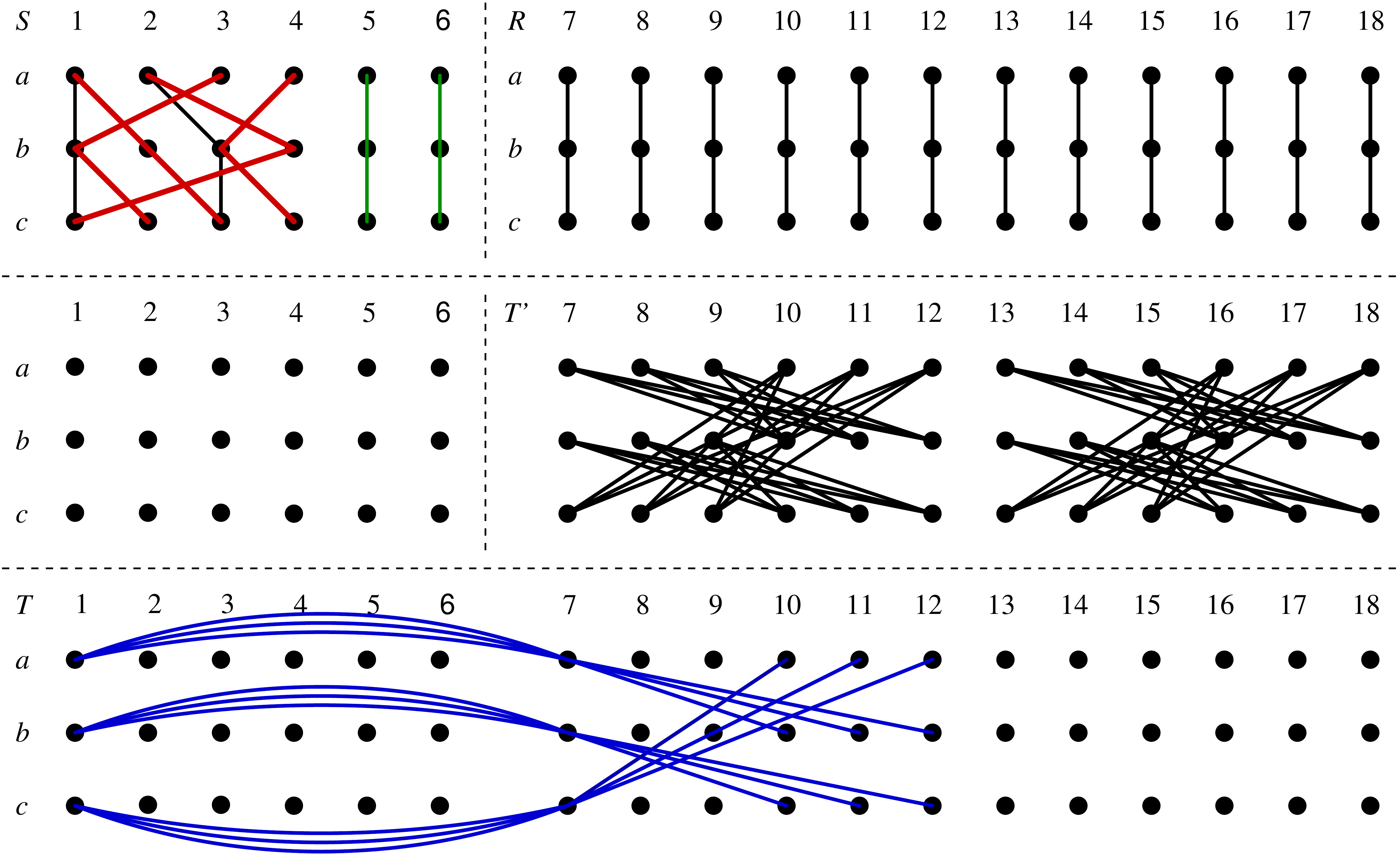}
\caption{This graphic illustrates the transformation from {\sc 3DM} to {\sc 3D$k$M} for $k=4$ as explained in Example \ref{E01}. The drawing on the top left visualizes an instance with $6$ triples in $S$ that cover the elements $\{a_1,\ldots,a_4,b_1,\ldots,b_4,c_1,\ldots,c_4\}$. The triples are indicated by 6 red and 2 black lines, each covering 3 elements. Set $S$ contains a matching indicated by the red lines. Each set $A$, $B$ and $C$ is extended by two element $a_5,a_6$, $b_5,b_6$ and $c_5,c_6$ respectively, and set $S$ is extended by two triples $(a_5,b_5,c_5),(a_6,b_6,c_6)$, such that the number of elements in the new sets $A$, $B$ and $C$ is a multiple of $(k-1)=3$. These two triples are indicated by green lines. The drawing in the middle right visualizes the $2 \cdot 6 = 12$ triples of $R$ indicated by black lines. The drawing at the bottom visualizes the $54$ tuples of $T' = T_{7,3} \cup T_{13,3}$, also indicated by black lines, each covering 2 elements. The set $T$ is formed from set $T'$ by adding each element of $A$, $B$ and $C$ to $k-1=3$ tuples of $T'$. For the sake of clarity, only the triples from T for the elements $a_1$, $b_1$ and $c_1$ are shown in the figure. These triples are indicated by blue lines.}
\label{F01}
\end{figure}

\begin{example}
\label{E01}
Let $A=\{a_1,\ldots,a_4\}$, $B=\{b_1,\ldots,b_4\}$, $C=\{c_1,\ldots,c_4\}$ and $$S= \{(a_1,b_1,c_1), (a_1,b_2,c_3), (a_2,b_3,c_3), (a_2,b_4,c_1), (a_3,b_1,c_2), (a_4,b_3,c_4)\}$$
be an instance for {\sc 3DM}. The triple $(a_1,b_2,c_3), (a_2,b_4,c_1), (a_3,b_1,c_2), (a_4,b_3,c_4)$ form a 3-di\-men\-sional matching and thus a solution for {\sc 3DM}.

It follows the construction of an instance for {\sc 3D$k$M} for $k = 4$ as defined in the proof of Theorem \ref{T01}. Integer $n$ has to be a multiple of $k-1 = 3$. To ensure this, $A$ is extended by $a_5$ and $a_6$, $B$ is extended by $b_5$ and $b_6$, $C$ is extended by $c_5$ and $c_6$ and $S$ is extended by $(a_5,b_5,c_5)$ and $(a_6,b_6,c_6)$. Now $n=6$ and $r = \frac{n}{k-1} = 2$.

Then $A'=\{a_1,\ldots,a_{18}\}$, $B'=\{b_1,\ldots,b_{18}\}$, $C'=\{c_1,\ldots,c_{18}\}$ and $R=\{(a_i,b_i,c_i) \, | \, i=7,\ldots,18\}$. Set $T'$ is defined as $T' \, = \, T_{7, \, 3} \, \cup \, T_{13, \, 3}$.
Finally, set $S'$ is defined as
$$S'= S \, \cup \, R \, \cup \, T,$$
where, for example,{\small
$$T_{7,3} = \left\{
\setlength{\arraycolsep}{1pt}
\begin{array}{lllllllll}
(a_7,b_{10}), & (a_7,b_{11}), & (a_7,b_{12}), &
(a_8,b_{10}), & (a_8,b_{11}), & (a_8,b_{12}), &
(a_9,b_{10}), & (a_9,b_{11}), & (a_9,b_{12}), \\
(b_7,c_{10}), & (b_7,c_{11}), & (b_7,c_{12}), &
(b_8,c_{10}), & (b_8,c_{11}), & (b_8,c_{12}), &
(b_9,c_{10}), & (b_9,c_{11}), & (b_9,c_{12}), \\
(c_7,a_{10}), & (c_7,a_{11}), & (c_7,a_{12}), &
(c_8,a_{10}), & (c_8,a_{11}), & (c_8,a_{12}), &
(c_9,a_{10}), & (c_9,a_{11}), & (c_9,a_{12}) \\
\end{array}
\right\},$$
$$T_{13,3} = \left\{
\setlength{\arraycolsep}{1pt}
\begin{array}{lllllllll}
(a_{13},b_{16}), & (a_{13},b_{17}), & (a_{13},b_{18}), &
(a_{14},b_{16}), & (a_{14},b_{17}), & (a_{14},b_{18}), &
(a_{15},b_{16}), & (a_{15},b_{17}), & (a_{15},b_{18}), \\
(b_{13},c_{16}), & (b_{13},c_{17}), & (b_{13},c_{18}), &
(b_{14},c_{16}), & (b_{14},c_{17}), & (b_{14},c_{18}), &
(b_{15},c_{16}), & (b_{15},c_{17}), & (b_{15},c_{18}), \\
(c_{13},a_{16}), & (c_{13},a_{17}), & (c_{13},a_{18}), &
(c_{14},a_{16}), & (c_{14},a_{17}), & (c_{14},a_{18}), &
(c_{15},a_{16}), & (c_{15},a_{17}), & (c_{15},a_{18}) \\
\end{array}
\right\},$$
$$T = \left\{
\setlength{\arraycolsep}{2pt}
\begin{array}{llllll}
(a_1,b_7,c_{10}), & (a_1,b_7,c_{11}), & (a_1,b_7,c_{12}), &
(a_2,b_8,c_{10}), & (a_2,b_8,c_{11}), & (a_2,b_8,c_{12}), \\
(a_3,b_9,c_{10}), & (a_3,b_9,c_{11}), & (a_3,b_9,c_{12}), &
(a_4,b_{13},c_{16}), & (a_4,b_{13},c_{17}), & (a_4,b_{13},c_{18}), \\
(a_5,b_{14},c_{16}), & (a_5,b_{14},c_{17}), & (a_5,b_{14},c_{18}), &
(a_6,b_{15},c_{16}), & (a_6,b_{15},c_{17}), & (a_6,b_{15},c_{18}), \\\\
(a_7,b_{10},c_1), & (a_7,b_{11},c_1), & (a_7,b_{12},c_1), &
(a_8,b_{10},c_2), & (a_8,b_{11},c_2), & (a_8,b_{12},c_2), \\
(a_9,b_{10},c_3), & (a_9,b_{11},c_3), & (a_9,b_{12},c_3), &
(a_{10},b_1,c_7), & (a_{10},b_2,c_8), & (a_{10},b_3,c_9), \\
(a_{11},b_1,c_7), & (a_{11},b_2,c_8), & (a_{11},b_3,c_9), &
(a_{12},b_1,c_7), & (a_{12},b_2,c_8), & (a_{12},b_3,c_9), \\\\
(a_{13},b_{16},c_4), & (a_{13},b_{17},c_4), & (a_{13},b_{18},c_4), &
(a_{14},b_{16},c_5), & (a_{14},b_{17},c_5), & (a_{14},b_{18},c_5), \\
(a_{15},b_{16},c_6), & (a_{15},b_{17},c_6), & (a_{15},b_{18},c_6), &
(a_{16},b_4,c_{13}), & (a_{16},b_5,c_{14}), & (a_{16},b_6,c_{15}), \\
(a_{17},b_4,c_{13}), & (a_{17},b_5,c_{14}), & (a_{17},b_6,c_{15}), &
(a_{18},b_4,c_{13}), & (a_{18},b_5,c_{14}), & (a_{18},b_6,c_{15})
\end{array}
\right\},$$}
see also Figure \ref{F01}.
\end{example}

\begin{theorem}
{\sc $k$-MD} is NP-complete for bipartite graphs $G$ and each $k \geq 2$.
\end{theorem}

\begin{proof}\label{theorem-kMD}
The {$k$-MD} problem is obviously in NP, because it can be checked in polynomial time whether a set of vertices is a $k$-resolving set.

The NP-hardness is proven by a reduction from {\sc 3D($k$-1)M}. Let $A = \{a_1,...,a_n\}$, $B = \{b_1,...,b_n\}$, $C = \{c_1,...,c_n\}$, $S = \{s_1,...,s_m\}$ be an instance $I$ for {\sc 3D$(k-1)$M} where $k \geq 2$ and $n > k$. The aim is to define a graph $G=(V,E)$ and a number $x$ such that $G$ has a $k$-resolving set of size $x$ if and only if instance $I$ has a ($k$-1)-matching.

Graph $G$ is defined as follows, see also Figure \ref{F02}. It has a vertex $a_i$, $b_i$ and $c_i$ for $i=1,\ldots,n$ and a vertex $s_i$ for $i=1,\ldots,m$. Graph $G$ additionally contains vertices denoted by $a_0, b_0, c_0, v_0,v_A, v_B, v_C$ and $d_1,\ldots,d_{m'}$ where $m' = \lceil log(m)\rceil$. 
\begin{enumerate}
\item
Each vertex $a_i$, $0 \leq i \leq n$, is connected with
\begin{enumerate}
\item
vertex $v_A$,
\item
vertex $v_0$, and
\item vertex $s_j$, $1 \leq j \leq m$ if and only if triple $s_j$ contains element $a_i$.
\end{enumerate}
\item
Each vertex $b_i$, $0 \leq i \leq n$, is connected with
\begin{enumerate}
\item
vertex $v_B$,
\item
vertex $v_0$, and
\item
vertex $s_j$, $1 \leq j \leq m$ if and only if triple $s_j$ contains element $b_i$.
\end{enumerate}
\item
Each vertex $c_i$, $0 \leq i \leq n$, is connected with
\begin{enumerate}
\item
vertex $v_C$,
\item
vertex $v_0$, and
\item
vertex $s_j$, $1 \leq j \leq m$ if and only if triple $s_j$ contains element $c_i$.
\end{enumerate}
\item
Each vertex $d_i$, $1 \leq i \leq {m'}$, is connected with
\begin{enumerate}
\item
vertex $v_0$ and
\item
vertex $s_j$, $1 \le j \le m$, if and only if the $i$-th bit of the binary representation of $j$ is $1$.
\end{enumerate}
\end{enumerate}
Graph $G$ contains additionally so-called {\em leg vertices}. These leg vertices form paths ({\em legs }) with $\lceil k/2 \rceil$ or $\lfloor k/2 \rfloor$ vertices. Two such legs, one with $\lceil k/2 \rceil$ vertices and one with $\lfloor k/2 \rfloor$ vertices, are attached to each vertex of $L_{\text{root}} = \{v_A, v_B, v_C, v_0, d_1,\ldots,d_{m'}\}$, see Figure \ref{F02}. Set $L_{\text{root}}$ is the set of {\em root vertices} of the legs. Let $L_{v}$ be the set of vertices of the two legs at vertex $v$ and
$$L = L_{v_A} \, \cup \, L_{v_B} \, \cup \, L_{v_C} \, \cup \, L_{v_0} \, \cup \, L_{d_1} \, \cup \, \cdots  \, \cup \, L_{d_{m'}}$$ be the set of all leg vertices of $G$. Set $L_{\text{root}}$ has $4+m'$ vertices, each set $L_v$, $v \in L_{\text{root}}$, has $k$ vertices and $L$ has $(4+m') k$ vertices.

The graph $G$ can obviously be constructed in polynomial time from instant $I$.

First of all, let us note some properties of G.
\begin{enumerate}
\item [P1:]
$G$ is bipartite.
\item [P2:]
The distance between
\begin{enumerate}
\item
two vertices of $\{v_B, v_B, v_C\}$ is 4,
\item
two vertices of $\{d_1,\ldots,d_{m'}\}$ is 2,
\item
a vertex of $\{v_B, v_B, v_C\}$ and a vertex of $\{d_1,\ldots,d_{m'}\}$ is 3,
\item
vertex $v_0$ and a vertex of $\{v_B, v_B, v_C\}$ is 2, and
\item
vertex $v_0$ and a vertex of $\{d_1,\ldots,d_{m'}\}$ is 1.
\end{enumerate}
\item [P3:]
Every $k$-resolving set for $G$ contains all vertices of $L$. This follows from the observation that for each vertex $v \in L_{\text{root}}$ the two vertices of $L_v$ adjacent with $v$ are only resolved by the $k$ vertices of $L_v$.
\end{enumerate}

Now we will prove that $S$ has a ($k$-1)-matching for instance $I$ if and only if $G$ has a resolving set of size $$x=(4+m')k + 3 + (k-1)n.$$ 

\medskip
"$\Rightarrow:$"
Let $M\subseteq S$ be a ($k$-1)-matching for instance $I$. The aim is to show that 
$$R \, = \, L \, \cup \, \{a_0, b_0, c_0\} \, \cup \, M$$
is a $k$-resolving set for $G$ of size $$x=(4+m')k + 3 + (k-1)n,$$
that is, each pair of two distinct vertices $u_1,u_2$ of $G$ is resolved by at least $k$ vertices of $U$. Here the triple $s_j$ of $M$ are considered as vertices of $G$.

Consider the following case distinctions for two vertices $u_1$ and $u_2$.
\begin{enumerate}
\item $u_1, u_2 \in L_{v}$, $v \in L_{\text{root}}$.
\begin{enumerate}
\item $d(u_1,v) = d(u_2,v)$. Each of the $k$ vertices of $L_{v}$ resolves $u_1$ and $u_2$.
\item $d(u_1,v) \not= d(u_2,v)$. Each of the $k$ vertices of $L_{v'}$, $v' \in L_{\text{root}} \setminus \{v\}$, resolves $u_1$ and $u_2$.
\end{enumerate}
\item $u_1 \in L_{v_1}$, $u_2 \in L_{v_2}$, $v_1, v_2 \in L_{\text{root}}$, $v_1 \not= v_2$, and $d(u_1,v_1) \leq d(u_2,v_2)$. Each of the $k$ vertices of $L_{v_1}$ resolves $u_1$ and $u_2$.
\end{enumerate}

Up to this point all pairs of vertices $u_1,u_2$ are considered of which both are in $L$.
\begin{enumerate}
\setcounter{enumi}{2}
\item $u_1 \in L_{v_A} \cup L_{v_B} \cup L_{v_C}$ and $u_2 \not\in L$. Each of the $k$ vertices of $L_{v_0}$ resolves $u_1$ and $u_2$.
\item $u_1 \in L_{d_1} \cup \cdots \cup L_{d_{m'}}$ and $u_2 \not\in L$.
\begin{enumerate}
\item $u_2 \not\in \{v_B,v_C\}$. Each of the $k$ vertices of $L_{v_A}$ resolves $u_1$ and $u_2$.
\item $u_2 \not\in \{v_A,v_C\}$. Each of the $k$ vertices of $L_{v_B}$ resolves $u_1$ and $u_2$.
\item $u_2 \not\in \{v_A,v_B\}$. Each of the $k$ vertices of $L_{v_C}$ resolves $u_1$ and $u_2$.
\end{enumerate}
\item $u_1 \in L_{v_0}$ and $u_2 \not\in L$.
\begin{enumerate}
\item
$u_2 \in \{v_A,a_0,\ldots,a_n\}$. Each of the $k$ vertices of $L_{v_A}$ resolves $u_1$ and $u_2$.
\item
$u_2 \in \{v_B,b_0,\ldots,b_n\}$. Each of the $k$ vertices of $L_{v_B}$ resolves $u_1$ and $u_2$.
\item
$u_2 \in \{v_C,c_0,\ldots,c_n\}$. Each of the $k$ vertices of $L_{v_C}$ resolves $u_1$ and $u_2$.
\item
$u_2 \in \{d_i\} \cup \{s_j \, | \, \text{the $i$-th bit in the binary representation of $j$ is $1$} \}$. Each of the $k$ vertices of $L_{d_i}$ resolves $u_1$ and $u_2$.
\end{enumerate}
\end{enumerate}

Up to this point all pairs of vertices $u_1,u_2$ are considered of which at least one of them is in $L$.
\begin{enumerate}
\setcounter{enumi}{5}
\item $u_1  \in L_{\text{root}}$ and $u_2 \not\in L$. Each of the $k$ vertices of $L_{u_1}$ resolves $u_1$ and $u_2$.
\end{enumerate}

Up to this point all pairs of vertices $u_1,u_2$ are considered of which at least one of them is in $L \cup L_{\text{root}}$.
\begin{enumerate}
\setcounter{enumi}{6}
\item $u_1 = s_{i_1} \in \{s_1,\ldots,s_{m'}\}$ and $u_2 \not\in L \cup L_{\text{root}}$.
\begin{enumerate}
\item $u_2 = s_{i_2} \in \{s_1,\ldots,s_{m'}\}$. Each of the $k$ vertices of $L_{d_j}$ resolves $u_1$ and $u_2$, if the binary representation of $i_1$ and $i_2$ differs in position $j$.
\item $u_2 \in  \{a_0,\ldots,a_n\}$, $u_2 \in  \{b_0,\ldots,b_n\}$, or $u_2 \in \{c_0,\ldots,c_n\}$. Each of the $k$ vertices of $L_{v_A}$, $L_{v_B}$, or $L_{v_C}$, respectively, resolves $u_1$ and $u_2$.
\end{enumerate}
\end{enumerate}

Up to this point all pairs of vertices $u_1,u_2$ are considered of which at least one of them is in $L \cup L_{\text{root}} \cup \{s_1,\ldots,s_{m'}\}$.
\begin{enumerate}
\setcounter{enumi}{7}
\item $u_1 \in  \{a_1,\ldots,a_n\}$ and $u_2 \not\in L \cup L_{\text{root}} \cup \{s_1,\ldots,s_{m'}\}$.
\begin{enumerate}
\item $u_2 \in \{b_0,\ldots,b_n,c_0,\ldots,c_n\}$. Each of the $k$ vertices of $L_{v_A}$ resolves $u_1$ and $u_2$.
\item $u_2 \in \{a_1,\ldots,a_n\}$. Each vertex $s_i$ for which triple $s_i$ contains $u_1$ or $u_2$ resolves $u_1$ and $u_2$. There are $2(k-1) \geq k$ such vertices for $k \geq 2$.
\item $u_2 = a_0$. Each vertex $s_i$ for which triple $s_i$ contains $u_1$ resolves $u_1$ and $u_2$, and vertex $a_0$ resolves $u_1$ and $u_2$. Altogether these are exactly $(k-1) + 1 = k$ vertices.
\end{enumerate}

\item $u_1 \in  \{b_1,\ldots,b_n\}$ and $u_2 \not\in L \cup L_{\text{root}} \cup \{s_1,\ldots,s_{m'}\}$. (as in case 8)

\item $u_1 \in  \{c_1,\ldots,c_n\}$ and $u_2 \not\in L \cup L_{\text{root}} \cup \{s_1,\ldots,s_{m'}\}$. (as in case 8)

\item $u_1, u_2 \in \{a_0,b_0,c_0\}$. Each of the $k$ vertices of $L_{v_A}$, $L_{v_B}$ or $L_{v_C}$ resolves $u_1$ and $u_2$.
\end{enumerate}

Now all pairs of vertices $u_1,u_2$ of $G$ are considered and it is shown that all of them are resolved by at least $k$ vertices from $R$. Note that only the vertex pairs $u_1,u_2 \in \{a_0,\ldots,a_n\}$, $u_1,u_2 \in \{b_0,\ldots,b_n\}$ and $u_1,u_2 \in \{c_0,\ldots,c_n\}$ are not already resolved by $k$ vertices of $L$. Strictly speaking, not a single vertex from $L \cup \{v_A,v_B,v_C,v_0,d_1,\ldots,d_{m'}\}$ resolves such a pair of vertices.

\medskip
"$\Leftarrow:$"
Let $R\subseteq V$ be a $k$-resolving set for $G$ with $x=(4+m')k + 3 + (k-1)n$ vertices. By Property P3, $R$ contains all the $(4+m)'k$ vertices of $L$. This leaves $3 + (k-1) n$ vertices of $R$ that are not in $L$. Let us now consider the vertex pairs $a_0,a_i$, and $b_0,b_i$, and $c_0,c_i$ for $i = 1,\ldots, n$. The vertices of $L$ and the vertices of $\{v_A,v_B,v_C,v_0,d_1,\ldots,d_{m'}\}$ do not resolve these vertex pairs. The only way to resolve these $3n$ vertex pairs at least $k$ times with $3+(k-1)n$ vertices for $n > k \geq 2$, is to use $k$-$1$ vertices from $\{s_1,\ldots,s_m\}$ that form a $k$-$1$ matching and the three vertices $a_0, b_0, c_0$. This is the point where it is necessary that $n$ is greater than $k$. 
\end{proof}

In the introduction of this paper, we mentioned that the {\sc $k$-Metric Dimension} and the {\sc $(k, t)$-Metric Dimension} in \cite{EYR16} are the same if $t$ is set to the diameter of $G$. Since the constructed graph in Theorem \ref{theorem-kMD} has diameter $2 \cdot \lceil k/2 \rceil + 3$, Theorem \ref{theorem-kMD} also proves the NP-completeness of {\sc $(k, t)$-Metric Dimension} for bipartite graphs, each $k \geq 2$ and $t \geq 2 \cdot \lceil k/2 \rceil + 3$.

\begin{figure}
\center
\includegraphics[width=0.90\textwidth]{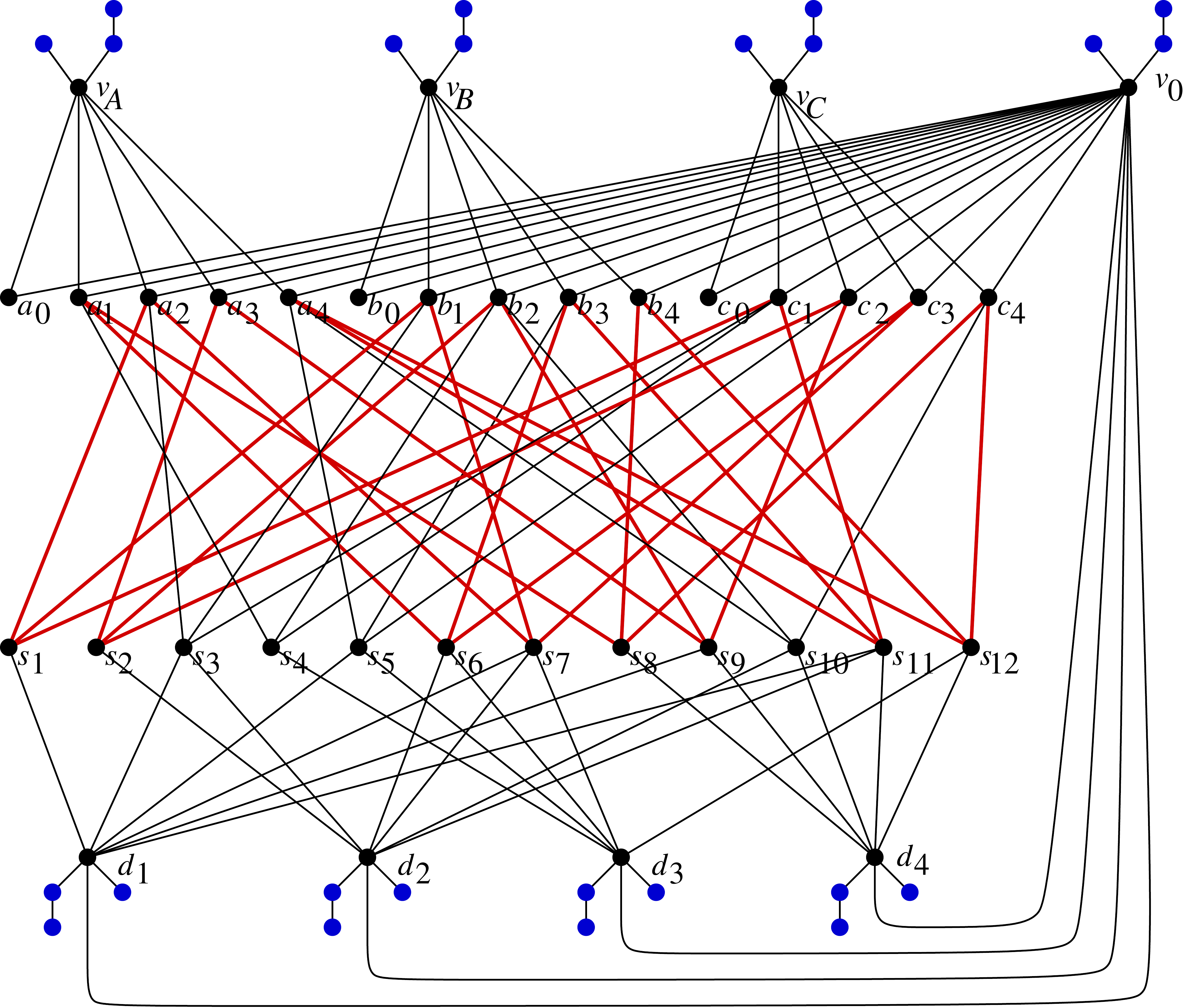}
\caption{This graphic illustrates the transformation from {\sc 3D2M} to {\sc 3-MD}. The Instance $I$ consisting of $A=\{a_1,\ldots,a_4\}$, $B=\{b_1,\ldots, b_4\}$, $C=\{c_1, \ldots, c_4\}$, $S=\{s_1,\ldots,s_{12}\}$ with $s_1=(a_2,b_1,c_1)$, $s_2=(a_3,b_2,c_2)$, $s_3=(a_2,b_1,c_1)$, $s_4=(a_1,b_2,c_1)$, $s_5 = (a_4,b_3,c_2)$, $s_6 = (a_1,b_3,c_3)$, $s_7=(a_2,b_1,c_3)$, $s_8=(a_1,b_4,c_4)$, $s_9=(a_3,b_2,c_2)$, $s_{10}=(a_4,b_2,c_4)$, $s_{11} = (a_4,b_3,c_1)$, $s_{12} = (a_4,b_4,c_4)$ for {\sc 3D2M} is transformed into the graph $G$ and $x=(4+4)3+3+(3-1)n=35$. The set of triples $M=\{s_1,s_2,s_6,s_7,s_8,s_9,s_{11},s_{12}\}$, indicated in the figure by the red lines, is a $2$-matching for instance $I$, where $L \cup \{a_0,b_0,c_0\} \cup M$ is a $3$-resolving set for $G$ of size $x$. Set $L$ is the set of vertices of the legs attached at the vertices $v_A,v_B,v_C,v_0,d_1,d_2,d_3,d_4$. In the figure, the vertices of $L$ are colored blue.
\label{F02}}
\end{figure}

%------------------------------------------------------------------------------
%\section{Conclusion}
%------------------------------------------------------------------------------

%%%%%%%%%%%%%%%%%%%%%%%%%%%%%%%%%%%%%%%%%%%%%%%%%%%%%%%%%%%%%%%%%%%%%%%%%%
\newcommand{\etalchar}[1]{$^{#1}$}

%%%%%%%%%%%%%%%%%%%%%%%%%%%%%%%%%%%%%%%%%%%%%%%%%%%%%%%%%%%%%%%%%%%%%%%%%%

%%%%%%%%%%%%%%%%%%%%%%%%%%%%%%%%%%%%%%%%%%%%%%%%%%%%%%%%%%%%%%%%%%%%%%%%%%

\begin{thebibliography}{EMYRV16}

% this bibliography is generated by alphadin.bst [8.2] from 2005-12-21

\providecommand{\url}[1]{\texttt{#1}}
\expandafter\ifx\csname urlstyle\endcsname\relax
  \providecommand{\doi}[1]{doi: #1}\else
  \providecommand{\doi}{doi: \begingroup \urlstyle{rm}\Url}\fi

\bibitem[BBS{\etalchar{+}}11]{BBSSS11}
\textsc{Ba{\v{c}a}}, Martin ; \textsc{Baskoro}, Edy~T. ; \textsc{Salman}, A.
  N.~M. ; \textsc{Saputro}, Suhadi~W.  ; \textsc{Suprijanto}, Djoko:
\newblock The Metric Dimension of Regular Bipartite Graphs.
\newblock {In: }\emph{Bulletin mathématiques de la Société des sciences
  mathématiques de Roumanie} 54 (2011), Nr. 1, S. 15--28

\bibitem[BEE{\etalchar{+}}05]{BEEHHMR05}
\textsc{Beerliova}, Zuzana ; \textsc{Eberhard}, Felix ; \textsc{Erlebach},
  Thomas ; \textsc{Hall}, Alexander ; \textsc{Hoffmann}, Michael ;
  \textsc{Miha{\v{l}}{\'a}k}, Mat{\'u}{\v{s}}  ; \textsc{Ram}, L.~S.:
\newblock Network Discovery and Verification.
\newblock {In: }\textsc{Kratsch}, Dieter (Hrsg.): \emph{Graph-Theoretic
  Concepts in Computer Science}, Springer Berlin Heidelberg, 2005, 127--138

\bibitem[CEJO00]{CEJO00}
\textsc{Chartrand}, Gary ; \textsc{Eroh}, Linda ; \textsc{Johnson}, Mark~A.  ;
  \textsc{Oellermann}, Ortrud:
\newblock Resolvability in graphs and the metric dimension of a graph.
\newblock {In: }\emph{Discrete Applied Mathematics} 105 (2000), Nr. 1-3,
  99--113.
\newblock \url{http://dx.doi.org/10.1016/S0166-218X(00)00198-0}. --
\newblock DOI 10.1016/S0166--218X(00)00198--0

\bibitem[CGH08]{CGH08}
\textsc{Chappell}, Glenn~G. ; \textsc{Gimbel}, John~G.  ; \textsc{Hartman},
  Chris:
\newblock Bounds on the metric and partition dimensions of a graph.
\newblock {In: }\emph{Ars Combinatoria} 88 (2008)

\bibitem[CPZ00]{CPZ00}
\textsc{Chartrand}, Gary ; \textsc{Poisson}, Christopher  ; \textsc{Zhang},
  Ping:
\newblock Resolvability and the upper dimension of graphs.
\newblock {In: }\emph{Computers and Mathematics with Applications} 39 (2000),
  Nr. 12, S. 19--28

\bibitem[DPSL12]{DPSL12}
\textsc{D{\'{\i}}az}, Josep ; \textsc{Pottonen}, Olli ; \textsc{Serna},
  Maria~J.  ; \textsc{Leeuwen}, Erik~J.:
\newblock On the Complexity of Metric Dimension.
\newblock {In: }\emph{Algorithms - {ESA} 2012 - 20th Annual European Symposium,
  Ljubljana, Slovenia, September 10-12, 2012. Proceedings}, 2012, 419--430

\bibitem[EMRY13]{ERY13}
\textsc{Estrada-Moreno}, Alejandro ;
  \textsc{Rodr{\'{\i}}guez{-}Vel{\'{a}}zquez}, Juan~A.  ; \textsc{Yero},
  Ismael~G.:
\newblock The k-metric dimension of a graph.
\newblock {In: }\emph{Applied Mathematics {\&} Information Sciences} 9 (2013),
  12, Nr. 6, S. 2829--2840.
\newblock \url{http://dx.doi.org/10.12785/amis/090609}. --
\newblock DOI 10.12785/amis/090609

\bibitem[EMYRV16]{EYR16}
\textsc{Estrada-Moreno}, Alejandro ; \textsc{Yero}, IG  ;
  \textsc{Rodr{\'\i}guez-Vel{\'a}zquez}, JA:
\newblock On the (k, t)-metric dimension of graphs.
\newblock {In: }\emph{The Computer Journal}  (2016)

\bibitem[GJ79]{GJ79}
\textsc{Garey}, Michael~R. ; \textsc{Johnson}, David~S.:
\newblock \emph{Computers and Intractability: {A} Guide to the Theory of
  NP-Completeness}.
\newblock W. H. Freeman, 1979. --
\newblock ISBN 0--7167--1044--7

\bibitem[Hay77]{Hay17}
\textsc{Hayat}, Sakander:
\newblock Computing distance-based topological descriptors of complex chemical
  networks: New theoretical techniques.
\newblock {In: }\emph{Chemical Physics Letters} 688 (1977), Nr. 1, 51--58.
\newblock \url{http://dx.doi.org/10.1016/j.cplett.2017.09.055}. --
\newblock DOI 10.1016/j.cplett.2017.09.055

\bibitem[HM76]{HM76}
\textsc{Harary}, Frank ; \textsc{Melter}, Robert~A.:
\newblock On the metric dimension of a graph.
\newblock {In: }\emph{Ars Combinatoria} 2 (1976), S. 191--195

\bibitem[HMP{\etalchar{+}}05]{HMPSCP05}
\textsc{Hernando}, M.~C. ; \textsc{Mora}, Merc{\`{e}} ; \textsc{Pelayo},
  Ignacio~M. ; \textsc{Seara}, Carlos ; \textsc{C{\'{a}}ceres}, Jos{\'{e}}  ;
  \textsc{Puertas}, Mar{\'{\i}}a~Luz:
\newblock On the metric dimension of some families of graphs.
\newblock {In: }\emph{Electronic Notes in Discrete Mathematics} 22 (2005),
  129--133.
\newblock \url{http://dx.doi.org/10.1016/j.endm.2005.06.023}. --
\newblock DOI 10.1016/j.endm.2005.06.023

\bibitem[HMSW08]{HMSW08}
\textsc{Hernando}, M.~C. ; \textsc{Mora}, Merc{\`{e}} ; \textsc{Slater},
  Peter~J.  ; \textsc{Wood}, David~R.:
\newblock Fault-tolerant metric dimension of graphs.
\newblock {In: }\emph{Convexity in Discrete Structures} 5 (2008), S. 81--85

\bibitem[HSV12]{HSV12}
\textsc{Hauptmann}, Mathias ; \textsc{Schmied}, Richard  ; \textsc{Viehmann},
  Claus:
\newblock Approximation complexity of Metric Dimension problem.
\newblock {In: }\emph{Journal of Discrete Algorithms} 14 (2012), 214--222.
\newblock \url{http://dx.doi.org/10.1016/j.jda.2011.12.010}. --
\newblock DOI 10.1016/j.jda.2011.12.010

\bibitem[HW12]{HW12}
\textsc{Hoffmann}, Stefan ; \textsc{Wanke}, Egon:
\newblock Metric Dimension for Gabriel Unit Disk Graphs Is NP-Complete.
\newblock {In: }\emph{Algorithms for Sensor Systems, 8th International
  Symposium on Algorithms for Sensor Systems, Wireless Ad Hoc Networks and
  Autonomous Mobile Entities, {ALGOSENSORS} 2012, Ljubljana, Slovenia,
  September 13-14, 2012. Revised Selected Papers}, 2012, 90--92

\bibitem[IBSS10]{IBSS10}
\textsc{Iswadi}, H. ; \textsc{Baskoro}, Edy~T. ; \textsc{Salman}, A.N.M.  ;
  \textsc{Simanjuntak}, Rinovia:
\newblock The metric dimension of amalgamation of cycles.
\newblock {In: }\emph{Far East Journal of Mathematical Sciences (FJMS)} 41
  (2010), Nr. 1, S. 19--31

\bibitem[KRR96]{KRR96}
\textsc{Khuller}, Samir ; \textsc{Raghavachari}, Balaji  ; \textsc{Rosenfeld},
  Azriel:
\newblock Landmarks in Graphs.
\newblock {In: }\emph{Discrete Applied Mathematics} 70 (1996), Nr. 3, 217--229.
\newblock \url{http://dx.doi.org/10.1016/0166-218X(95)00106-2}. --
\newblock DOI 10.1016/0166--218X(95)00106--2

\bibitem[LA06]{LA06}
\textsc{Liu}, Ke ; \textsc{Abu{-}Ghazaleh}, Nael~B.:
\newblock Virtual Coordinates with Backtracking for Void Traversal in
  Geographic Routing.
\newblock {In: }\emph{Ad-Hoc, Mobile, and Wireless Networks, 5th International
  Conference, {ADHOC-NOW} 2006, Ottawa, Canada, August 17-19, 2006,
  Proceedings}, 2006, 46--59

\bibitem[MT84]{MT84}
\textsc{Melter}, Robert~A. ; \textsc{Tomescu}, Ioan:
\newblock Metric bases in digital geometry.
\newblock {In: }\emph{Computer Vision, Graphics, and Image Processing} 25
  (1984), Nr. 1, 113--121.
\newblock \url{http://dx.doi.org/10.1016/0734-189X(84)90051-3}. --
\newblock DOI 10.1016/0734--189X(84)90051--3

\bibitem[Sla75]{Sla75}
\textsc{Slater}, Peter~J.:
\newblock Leaves of trees.
\newblock {In: }\emph{Congressum Numerantium} 14 (1975), S. 549--559

\bibitem[ST04]{ST04}
\textsc{Seb{\"{o}}}, Andr{\'{a}}s ; \textsc{Tannier}, Eric:
\newblock On Metric Generators of Graphs.
\newblock {In: }\emph{Mathematics of Operations Research} 29 (2004), Nr. 2,
  383--393.
\newblock \url{http://dx.doi.org/10.1287/moor.1030.0070}. --
\newblock DOI 10.1287/moor.1030.0070

\bibitem[YER17]{YER17}
\textsc{Yero}, Ismael~G. ; \textsc{Estrada{-}Moreno}, Alejandro  ;
  \textsc{Rodr{\'{\i}}guez{-}Vel{\'{a}}zquez}, Juan~A.:
\newblock Computing the k-metric dimension of graphs.
\newblock {In: }\emph{Applied Mathematics and Computation} 300 (2017), 60--69.
\newblock \url{http://dx.doi.org/10.1016/j.amc.2016.12.005}. --
\newblock DOI 10.1016/j.amc.2016.12.005

\end{thebibliography}
\end{document}